\documentclass{llncs}

\usepackage{times}
\usepackage{amsmath}
\usepackage{amssymb}
\usepackage{amsfonts}
\usepackage{color}
\usepackage{caption}
\usepackage{calc}
\usepackage{graphicx}
\usepackage{bm}
\usepackage[lined,boxed]{algorithm2e}
\usepackage{multirow}

\definecolor{darkred}{RGB}{127,0,0}
\definecolor{darkgreen}{RGB}{0,127,0}
\definecolor{darkblue}{RGB}{0,0,127}

\newcommand{\RE}{\mathbb{R}}
\newcommand{\REp}{\mathbb{R}^d_{+}}
\newcommand{\REpp}{\mathbb{R}^d_{++}}

\newcommand{\OBJ}{{\mathcal{O}}}
\newcommand{\WPP}{\succsim}
\newcommand{\PP}{\succ}

\newcommand{\EFF}{\text{EFF}}

\newcommand{\WST}{\text{WST}}
\newcommand{\PE}{\text{PN}}

\newcommand{\CA}{\mathcal{A}}
\newcommand{\CE}{\mathcal{E}}
\newcommand{\CF}{\mathcal{F}}

\newcommand{\CC}{\mathcal{C}}

\usepackage{makeidx}  
\begin{document}
\mainmatter              
\title{Efficiency in Multi-objective Games}
\titlerunning{Efficiency in Multi-objective Games}  
%
\author{Anisse Ismaili$^{1,2}$}
%
\authorrunning{Anisse Ismaili} 
%
\institute{$^1$ Universit\'{e} Pierre et Marie Curie,
Univ Paris 06, UMR 7606, LIP6, F-75005, Paris, France\\
$^2$ Paris Dauphine University, Place du Mal de Lattre de Tassigny, 
75775 Paris Cedex 16, France\\
\email{anisse.ismaili@lip6.fr}}

\maketitle              

\vspace*{-0.5cm}
\begin{abstract}
In a multi-objective game, each agent individually evaluates each overall action-profile on multiple objectives. I generalize the price of anarchy to multi-objective games and provide a polynomial-time algorithm to assess it$^1$.\\
This work asserts that policies on tobacco promote a higher economic efficiency.
\end{abstract}
\vspace*{-0.8cm}


\section{Introduction}

Economic agents, for each individual decision, make a trade off between multiple objectives, like for instance: time, resources, goods, financial income, sustainability, happiness and life. This motivated the introduction of a super-class of games: multi-objective (MO) games \cite{blackwell1956analog,shapley1959equilibrium}. Each agent evaluates each overall action profile by a \emph{vector}. His individual preference is a \emph{partial} rationality modelled by the Pareto-dominance. It induces Pareto-Nash-equilibria (PN) as the overall selfish outcomes. 
Furthermore, concerning economic models, such vectorial evaluations are a humble backtrack from the intrinsic and subjective theories of value, towards a non-theory of value where the evaluations are maintained vectorial, in order to enable partial rationalities and to avoid losses of information in the model. 
In this more realistic (behaviourally less assumptive) framework, in order to avoid critical losses of information on the several objectives in the model, thoroughly computing  efficiency  is a tremendous necessity \cite{madeley1999big,sloan2004price}. 

The literature on MO games is disparate and will be presented where relevant.
After the preliminaries below, Section \ref{sec:mopoa} generalizes the \textit{coordination ratio} (CR, better known as ``price of anarchy'') to MO games. Section \ref{sec:application} applies it to the efficiency of tobacco economy. Section \ref{sec:computation} provides algorithms\footnote{Appendix \ref{app:smooth} shows that ``smoothness'' analysis \cite{roughgarden2009intrinsic} cannot be applied to MO games.} to assess the MO-CR.  

Let $N=\{1,\ldots,n\}$ denote the \textit{set of agents}.
Let $A^i$ denote each agent $i$'s \textit{action-set} (discrete, finite).
Each agent $i$ decides an \textit{action} $a^i\in A^i$.
Given a subset of agents $M\subseteq N$, let $A^M$ denote $\times_{i\in M} A^i$ and let $A=A^N$ denote the \textit{set of overall action-profiles}.
Let $\OBJ=\{1,,\ldots,d\}$ denote the \textit{set of all the objectives}, with $d$ fixed.
Let $v^i: A\rightarrow \REp$ denote an agent $i$'s \textit{individual MO evaluation function}, which maps each overall action-profile $a=(a^1,\ldots,a^n)\in A$ to an MO evaluation $v^i(a)\in\REp$. Hence, agent $i$'s evaluation for objective $k$ is $v^i_k(a)\in\RE_+$. Given an overall action-profile $a\in A$, $a^M$ is the restriction of $a$ to $A^M$, and $a^{-i}$ to $A^{N\setminus\{i\}}$. 

\begin{definition}\quad A Multi-objective Game (MOG) is a tuple $\left(N, \{A^i\}_{i\in N}, \OBJ, \{v^i\}_{i\in N}\right)$.
\end{definition}
For instance, MO games encompass single-objective (discrete) optimization problems, MO optimization problems and non-cooperative games.
Assuming $\alpha=|A^i|\in\mathbb{N}$ for each agent, the representation of an MOG requires $n\alpha^n$ $d$-dimensional vectors. 

Let us now supply the vectors with a preference relation. Assuming a \textit{maximization} setting, given $x,y\in\REp$, the following relations state respectively that $y$ (\ref{eq:wpp}) weakly-Pareto-dominates and (\ref{eq:pp}) Pareto-dominates $x$:
\begin{eqnarray}
y\WPP x &\hspace{0.5cm}\Leftrightarrow\hspace{0.5cm}& \forall k\in\OBJ,~~ y_k\geq x_k
\label{eq:wpp}\\
y\PP x &\Leftrightarrow & \forall k\in\OBJ,~~ y_k\geq x_k\text{~~and~~}\exists k\in\OBJ,~~ y_k> x_k
\label{eq:pp}
\end{eqnarray}
The Pareto-dominance is a \emph{partial} order, inducing a multiplicity of Pareto-efficient outcomes. Formally, the set of efficient vectors is defined as follows:
\begin{definition}[Pareto-efficiency] 
For $Y\subseteq\REp$, the efficient vectors $\EFF[Y]\subseteq Y$ are:
$$\EFF[Y]=\{y^\ast\in Y~~|~~ \forall y\in Y, \mbox{~not~} (y\PP y^\ast)\}$$
\end{definition}
(Similarly, let $\WST[Y]=\{y^{-}\in Y\mbox{~s.t.~} \forall y\in Y, \mbox{~not~} (y^{-}\PP y)\}$ denote the subset of worst vectors.)
Pareto-efficiency enables to define as efficient all the trade-offs that cannot be improved on one objective without being downgraded on another one, that is: the best compromises between objectives (see e.g. Figure \ref{fig:eff}). 

At the individual scale, Pareto-efficiency defines a \textit{partial rationality}, enabling to model behaviours that single-objective (SO) games would not model consistently.   
\begin{definition}[Pareto-Nash equilibrium \cite{shapley1959equilibrium}]\label{def:PE}
In an MOG, an action-profile $a\in A$ is a Pareto-Nash equilibrium (denoted by $a\in\PE$), if and only if, for each agent $i\in N$:
$$v^i(a^i,a^{-i})\quad\in\quad\EFF\left[\quad v^i(A^i,a^{-i})\quad\right]$$
where $v^i(A^i,a^{-i})$ denotes $\{v^i(b^i,a^{-i})\mid b^i\in A^i\}$.
\end{definition}
Pareto-Nash equilibria encompass most behaviourally possible action-profiles. For instance, whatever an agent's subjective linear positive weighted combination of the objectives, his decision is Pareto-efficient. One can distinguish behavioural objectives inducing $\PE$ and also objectives on which to focus an efficiency study.

{\it Equilibrium existence.}
In many sound probabilistic settings \cite{daskalakis2011connectivity,dresher1970probability,rinott2000number}, Pareto efficiency is not demanding on the conditions of individual rationality, hence there are multiple Pareto-efficient responses. Consequently, pure PN are numerous in average: $|\PE|\in\Theta(\alpha^{\frac{d-1}{d}n})$, justifying their existence in a probabilistic manner. Furthermore, in MO games with MO potentials \cite{monderer1996potential,patrone2007multicriteria,rosenthal1973class}, the existence is guaranteed. 

\begin{example}[A didactic toy-example in Ocean Shores]
Five shops in Ocean Shores (the nodes) can decide upon two activities: renting bikes or buggies, selling clams or fruits, etc. Each agent evaluates his local action-profile depending on the actions of his inner-neighbours and according to two objectives: financial revenue and sustainability.\\
\hspace*{-31cm}\includegraphics[scale=0.8]{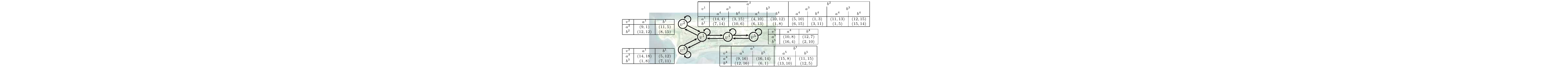}\\
For instance, we have $(b^1,b^2,a^3,b^4,b^5)\in\PE$, since each of these individual actions, given the adversary local action profile (column), is Pareto-efficient among the two actions of the agent (row). Even if the relative values of the objectives cannot be certainly ascertained, all the subjectively efficient vectors are encompassed by the individual Pareto-efficiency. In this MO game, there are $13$ Pareto-Nash-equilibria, which utilitarian evaluations are depicted in Figure \ref{fig:eff} (Section \ref{sec:mopoa}).
\label{ex:Ocean}
\end{example}

\section{The Multi-objective Coordination Ratio}
\label{sec:mopoa}

It is well known in game theory that an equilibrium  can be overall inefficient with regard to the sum of the individual evaluations. This loss of efficiency is measured by the {\it coordination ratio}\footnote{As Smoothness \cite{roughgarden2009intrinsic} cannot be applied to MO games, I cannot use the term \textit{Price of Anarchy}.} (CR) \cite{aland2006exact,awerbuch2005price,christodoulou2005price,guo2005price,koutsoupias1999worst,roughgarden2009intrinsic,roughgarden2007introduction}   $\min[u(PE)]/\max[u(A)]$. Regrettably, when focusing on one sole objective (e.g. making money or a higher GDP), there are losses of efficiency that are not measured (e.g. non-sustainability of productions or production of addictive carcinogens). This appeals for a more thorough analysis of the loss of efficiency at equilibrium and the definition of a {\it multi-objective} coordination ratio.
%
%

The utilitarian social welfare $u:A\rightarrow\REp$ is a vector-valued function  measuring social welfare with respect to the $d$ objectives: $u(a)=\sum_{i\in N} v^i(a)$, excluding the purely behavioural objectives that cause irrationality \cite{sloan2004price}.
Given a function $f:A\rightarrow Z$, the \textit{image set} $f(E)$ of a subset $E\subseteq A$ is defined by $f(E)=\{f(a)|a\in E\}\subseteq Z$. Given $\rho,y,z\in\REp$, the vector $\rho\star y\in\REp$ is defined by $\forall k\in\OBJ, (\rho\star y)_k=\rho_k y_k$ and the vector $y / z\in\REp$ is defined by $\forall k\in\OBJ, (y / z)_k=y_k / z_k$. 
For $x\in\REp$, $x\star Y$ denotes $\{x\star y\in\REp ~~|~~ y\in Y\}$. 
Given $x\in\REp$, $\CC(x)$ denotes $\{y\in\REp~|~x\WPP y\}$.
I also introduce\footnote{To enable ratios, one can do the minor assumption $\CF\subseteq\REpp$.} the notations $\CE$ and $\CF$, illustrated in Figures \ref{fig:eff} and \ref{fig:mopoa}:

\noindent
\begin{minipage}{0.58\columnwidth}
\noindent 
\begin{itemize}
\item \textcolor{darkblue}{$\CA=u(A)$} the set of \textit{outcomes}. $\textcolor{darkblue}{(\bullet)}$ 
\item \textcolor{darkgreen}{$\CE=u(\PE)$} the \textit{equilibria outcomes}. $\textcolor{darkgreen}{(\blacklozenge)}$
\item \textcolor{darkred}{$\CF=\EFF[u(A)]$} the \textit{efficient outcomes}. $\textcolor{darkred}{(\times)}$
\end{itemize}
For SO games, the worst-case efficiency of equilibria is measured by the CR $\min[u(PE)]/\max[u(A)]$.  However, for MO games, there are many equilibria and optima, and a ratio of the (green) \textit{set} \textcolor{darkgreen}{$\CE$} over the (red) \textit{set} \textcolor{darkred}{$\CF$} is not defined yet and ought to maintain the information on each objective without introducing dictatorial choices.
\end{minipage}
~~~
\begin{minipage}{0.38\columnwidth}
\centering 
%
\includegraphics[scale=0.36]{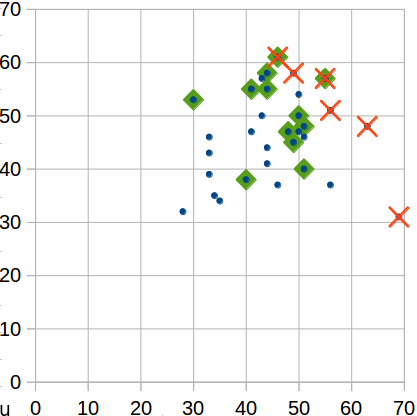}
\captionof{figure}{The bi-objective utilitarian vectors of Ocean Shores}
\label{fig:eff}
\end{minipage}\\[0.5ex]

I introduce a multi-objective CR.
Firstly, the efficiency of one equilibrium $y\in\CE$ is quantified without taking side with any efficient outcome, by defining with flexibility and no dictatorship, a \textit{disjunctive set} of guaranteed ratios of efficiency $R[y,\CF]=\bigcup_{z\in\CF}\CC(y/z)$.
Secondly, in MOGs, in average, there are many Pareto-Nash-equilibria. An efficiency \textit{guarantee} $\rho\in\REp$, must hold \textit{for each} equilibrium-outcome, inducing the conjunctive definition of the set of guaranteed ratios $R[\CE,\CF]=\bigcap_{y\in\CE} R[y,\CF]$.
Technically, $R[\CE,\CF]$ only depends on $\WST[\CE]$ and $\CF$. 
Finally, if two bounds on the efficiency $\rho$ and $\rho'$ are such that $\rho\PP\rho'$, then $\rho'$ brings no more information, hence, MO-CR is defined by using $\EFF$ on the guaranteed efficiency ratios $R[\WST[\CE],\CF]$.
This MO-CR satisfies a set of key properties detailed in Appendix \ref{app:axioms}.

\begin{definition}[MO Coordination Ratio]\label{def:mopoa}
Given an MOG, a vector $\rho\in\REp$ bounds the MOG's inefficiency if and only if it holds that:\quad
$\forall y\in\CE,\quad \exists z\in\CF,\quad y/z\WPP \rho$.
Consequently, the set of guaranteed ratios is defined by: 
$$R[\CE,\CF]\quad=\quad\bigcap_{y\in\CE}\bigcup_{z\in\CF}\CC(y/z)$$
and the MO-CR is defined by:\quad
$
\text{MO-CR}[\CE,\CF]=\EFF[R[\WST[\CE],\CF]]
$
\end{definition}
%
%
%
%


\begin{example}[The Efficiency ratios of Example \ref{ex:Ocean}]\label{ex:mopoashores}
I depict the efficiency ratios of Ocean Shores (intersected with $[0,1]^d$) which depend of $\WST[\CE]=\{(30,53), (40,38)\}$ and $\CF=\{(46,61),\ldots, (69,31)\}$. 
The part below the red line corresponds to $R[(30,53),\CF]$, the part below the blue line to $R[(40,38),\CF]$ and the yellow part below both lines is the conjunction on both equilibria $R[\WST[\CE],\CF]$. 
The freedom degree of deciding   what 

\vspace*{0.1cm}
\noindent
\begin{minipage}{0.55\columnwidth}
the overall efficiency should be is left free (no dictatorship) which results in several ratios in the MO-CR. 
Firstly, for each $\rho\in R[\CE,\CF]$, we have $\rho_1\leq 65\%$. Hence, whatever the choices of overall efficiency, one cannot guarantee more than \textit{65\% of efficiency on objective 1}.
Secondly, there are some subjectivities for which the efficiency on objective 2 is already total (100\%, if not more) while situation on objective 1 is worse and only 50\% can be obtained.
Thirdly, from 50\% to 65\% of subjective efficiency on objective 1, the various subjectivities range the  efficiency on objective 2 from 100\% to 75\%.
\end{minipage}
~~
\begin{minipage}{0.43\columnwidth}
\includegraphics[scale=0.44]{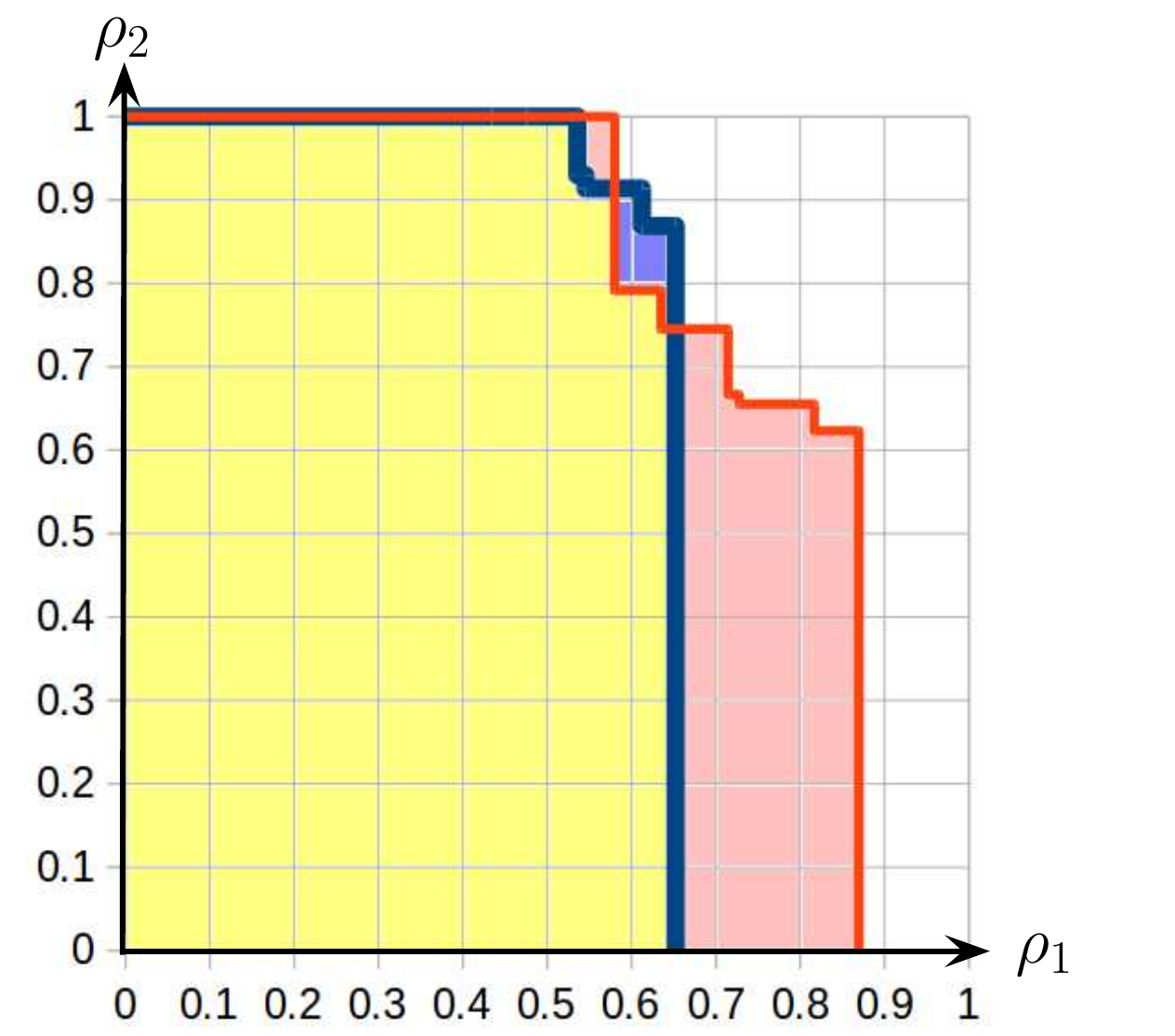}
\captionof{figure}{The MO-CR of Ocean Shores}
\end{minipage}
\end{example}

\vspace*{1cm}
\noindent
\begin{minipage}{0.46\columnwidth}
\centering
\includegraphics[scale=0.8]{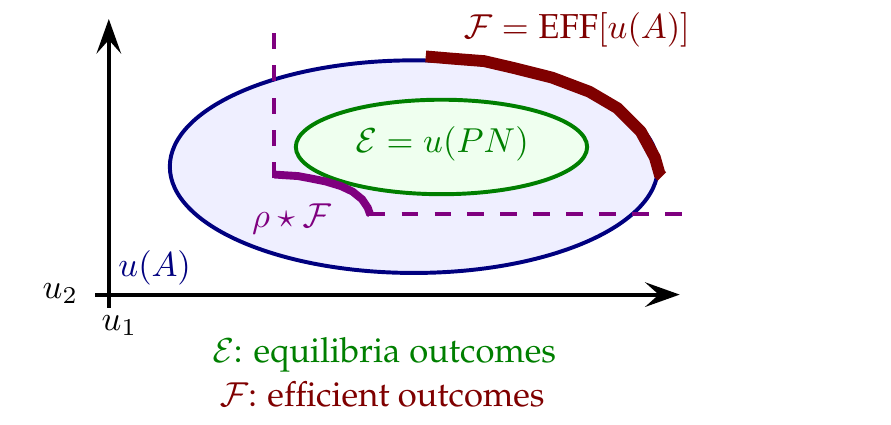}
\end{minipage}
~~~
\begin{minipage}{0.5\columnwidth}
Having $\rho$ in $\text{MO-CR}$ means that for each $y\in\CE$, there is an efficient outcome $z^{(y)}\in\CF$ such that $y$ dominates $\rho\star z^{(y)}$. In other words, if $\rho\in R[\CE,\CF]$, then each equilibrium satisfies the ratio of efficiency $\rho$. This means that equilibria-outcomes are at least as good as $\rho\star\CF$. That is: $\CE\subseteq(\rho\star\CF)+\REp$. Moreover, since $\rho$ is tight, $\CE$ sticks to $\rho\star\CF$.
\end{minipage}
\captionof{figure}{$\rho\in\text{MO-CR}$ bounds below $\CE$'s inefficiency:\quad $\CE~~\subseteq~~(\rho\star\CF)+\REp$}
\label{fig:mopoa}~\\

\section{Application to Tobacco Economy}\label{sec:application}

Tobacco consumption is a striking example of economic inefficiency induced by bounded rationalities.
According to the World Health Organisation \cite{world2011report}, 17.000 humans die each day of smoking related diseases (one person per 5 seconds). Meanwhile, addictive satisfaction and the financial revenue of the tobacco industry fosters consumption and production. 
According to the subjective theory of value \cite{walras1896elements}, some economists would say: ``Since consumers value the product, then the industry creates value.'' According to other health economists \cite{sloan2004price}, most consumers become addict before age 18, and as adults, would prefer a healthier life, but fail to opt-out.

\hspace*{0.5cm}The theory of MO games, based on a non-theory of value, just maintains vectorial evaluations and properly considers dollars, addiction and life expectancy as distinct objectives, with PN equilibria encompassing the relevant behaviours, even irrational. 

We \textbf{modelled} the tobacco industry and its consumers \cite{globalissuestobacco,madeley1999big} by a succinct MOG, with the help of (..) the association\footnote{I am grateful to Cl\'{e}mence Cagnat-Lardeau for her help on modelling tobacco economy.} ``\textit{Alliance contre le tabac}''. 
The set of agents is $N=\{\text{industry}, \nu \text{ consumers}\}$, where there are about $\nu=6.10^9$ prospective consumers.
Each consumer decides in $A^{\text{consumer}}=\{\text{not-smoking},\text{smoking}\}$ and cares about money, his addictive pleasure, and living. The industry only cares about money and decides in $A^{\text{industry}}=\{\text{not-active},\text{active},\text{advertise\&active}\}$. We have $\OBJ=\{\text{money},\text{reward},\text{life-expectancy}\}$. 
The tables below depict the evaluation vectors (over a life-time and ordered as in $\OBJ$) of one prospective consumer and the evaluations of the industry with respect to the number $\theta\in\{0,\ldots,\nu\}$ of consumers who decide to smoke. The money budget (already an aggregation) is expressed in kilo-dollars\footnote{Note that most states set the prices of tobacco, hence prices do not follow supply/demand.}\footnote{These numbers differ from \cite{sloan2004price} which aggregates everything (e.g. life expectancy) into money.}; the addictive reward is on an  ordinal scale $\{1,2,3,4\}$; life-expectancy is in years. 

\begin{center}
\begin{tabular}{r|ccc}
$v^{\text{consumer}}$ & not-active & active & advertise\&active\\
\hline
not-smoking &
$(48,1,75)$ & $(48,1,75)$ & $(48,1,75)$\\
smoking &
$(48,1,75)$ & $(12,3,65)$ & $(0,4,55)$\\
\multicolumn{4}{c}{}\\
$v^{\text{industry}}(\theta)$ & not-active & active & advertise\&active \\
\hline
$(\nu-\theta)\times$ &
$(0,-,-)$ & $(0,-,-)$ & $(0,-,-)$ \\
$+$ \hspace*{0.8cm}$\theta~~\times$ &
$(0,-,-)$ & $(26,-,-)$ & $(36,-,-)$ \\
\end{tabular}
\end{center}

\textbf{Pareto-Nash equilibria.} If the industry is active, then for the consumer, deciding to smoke or not depends on how the consumer subjectively values/weighs money, addiction and life expectancy: both decisions are encompassed by Pareto-efficiency. For the industry, advertise\&active is a dominant strategy. Consequently, Pareto-Nash-equilibria are all the action-profiles in which the industry decides advertise\&active.

\textbf{Efficiency.}
Since addiction is  irrational (detailed in Appendix \ref{app:tabac}), I focus on money and life-expectancy.
We have 
$
\CE
=
\{\theta(36,55)+(\nu-\theta)(48,75) \mid 0\leq\theta\leq \nu\}
$
and
$\CF=\{\nu(48,75)\}$,
where $\nu$ is the world's population, and $\theta$ the number of smokers. Since $\WST[\CE]=\{(36,55)\}$, the MO-CR is the singleton $\{(75\%,73\%)\}$: in the worst case, we lose 12k\$ and 20 years of life-expectancy per-consumer.  These Pareto-Nash-equilibria are the worst action-profiles for money and life-expectancy, a critical information that was not lost by this MOG and its MO-CR.

\textbf{Practical lessons.} Advertising tobacco fosters consumption. The association ``\textit{Alliance contre le tabac}'' passed a law for standardized neutral packets (April 3rd 2015), in order to annihilate all the benefits of branding, but only in France.  The model indicates that:
\begin{center}
\emph{This law will promote a higher economic efficiency}.
\end{center}

\section{Computation of the MO-CR}
\label{sec:computation}

In this section, I provide a polynomial-time algorithm  for the computation of MO-CR which relies on a very general procedure based on two phases:
\begin{enumerate}
\item Given a MOG, compute the worst equilibria $\WST[\CE]$ and the efficient outcomes $\CF$.
\item Given $\WST[\CE]$ and $\CF$, compute $\text{MO-CR}=\EFF[~R[~\WST[\CE]~,~\CF~]~]$.
\end{enumerate}
Depending on the input (normal form or compact representation), it adapts as follows.

\subsection{Computation of the MO-CR for Multi-objective Normal Forms}

For a MOG given in MO normal form (which representation length is $L=n\alpha^n d$), Phase 1 (computing $\WST[\CE]$ and $\CF$) is easy and takes time $O(L^2)$ (see Appendix \ref{app:proof}). For $d=2$, this lowers to $O(L\log_2(L))$. Let us denote the sizes of the outputs $q=|\WST[\CE]|$ and $m=|\CF|$. For normal forms, it holds that $q,m=O(|\CA|)=O(L)$.

\hspace*{0.5cm}For Phase 2, at first glance, the development of the intersection of unions $R[\WST[\CE],\CF]=\cap_{y\in\WST[\CE]}\cup_{z\in\CF}\CC(y/z)$ causes an exponential $m^q$. But fortunately, one can compute the $\text{MO-CR}$ in polynomial time. Below, $D^t$ is a set of vectors. Given two vectors $x,y\in\REp$, 
let $x\wedge y$ denote the vector defined by $\forall k\in\OBJ,~(x\wedge y)_k=\min\{x_k,y_k\}$ and recall that $\forall k\in\OBJ,~(x/y)_k=x_k/y_k$. Algorithm \ref{alg:givenEFoutputMOPoA:polynomial} is the development of  $\cap_{y\in\WST[\CE]}\cup_{z\in\CF}\CC(y/z)$, on a set-algebra of cone-unions. 
Appendix \ref{app:proof} shows that Algorithm \ref{alg:givenEFoutputMOPoA:polynomial} takes time $O((qm)^{2d-1}d)$, or $O((qm)^{2}\log_2(qm))$ for $d=2$. 

\vspace*{-0.4cm}
\restylealgo{ruled}
\begin{algorithm}[!h]
\KwIn{$\WST[\CE]=\{y^1,\ldots,y^q\}$ and $\CF=\{z^1,\ldots,z^m\}$}
\KwOut{$\text{MO-CR}=\EFF[R[\WST[\CE],\CF]]$}
\ \\ [-1.5ex]
{\bf create } $D^1\leftarrow \{y^1/z\in\REp~|~z\in\CF\}$\\
\For{$t=2,\ldots,q$}{
$D^t\leftarrow\EFF[\{\rho ~\wedge~ (y^t/z)~~|~~\rho\in D^{t-1},~~z\in\CF\}]$
}
{\bf return } $D^q$ 
\caption{Computing MO-CR in polynomial-time $O((qm)^{2d-1}d)$}
\label{alg:givenEFoutputMOPoA:polynomial}
\end{algorithm}
\vspace*{-0.4cm}

Having specified Phase 1 and 2 for normal forms, Theorem 1 follows:

\begin{theorem}[Computation of MO-CR]\label{th:master}
Given a MO normal form, one can compute the MO-CR in polynomial time $O(L^{4d-2})$.\quad If $d=2$, it lowers to $O(L^4\log_2(L))$.
\end{theorem}

\subsection{Computation of the MO-CR for Multi-objective Compact Representations}

Compact representations of massively multi-agent games (e.g. MO graphical games, MO action-graph games) have a representation length $L$ that is polynomial with respect to the number of agents $n$ and the sizes of the action-sets $\alpha$. As $q=|\WST[\CE]|$ and $m=|\CF|$ can be exponentials $\alpha^n$ of this representation length, compact representations are algorithmically more challenging,  leaving open the computation of $\WST[\CE]$ and $\CF$ in Phase 1, and complicating the use of Algorithm  \ref{alg:givenEFoutputMOPoA:polynomial} in Phase 2. 
To overcome this, one can do MO approximations \cite{papadimitriou2000approximability}, by implementing an approximate Phase 1 which precision transfers to Phase 2 in polynomial time, as follows.

\begin{lemma}\label{lem:approx} Given  $\varepsilon_1,\varepsilon_2>0$ and approximations $E$ of $\CE$ and $F$ of $\CF$ in the sense that:
\begin{eqnarray}
\forall y\in\CE, \exists y'\in E,\quad y\WPP y' &
\quad \text{and}\quad & 
\forall y'\in E, \exists y\in\CE,\quad (1+\varepsilon_1)y'\WPP y\label{eq:approx:E}\\
\forall z'\in F, \exists z\in\CF,\quad z'\WPP z &
\quad \text{and}\quad & 
\forall z\in\CF, \exists z'\in F,\quad (1+\varepsilon_2)z\WPP z'\label{eq:approx:F}
\end{eqnarray}
it holds that $R[E,F]\subseteq R[\CE,\CF]$ and:
\begin{eqnarray}
\forall \rho\in R[\CE,\CF], \exists \rho'\in R[E,F],\quad (1+\varepsilon_1)(1+\varepsilon_2)\rho'\WPP\rho\label{eq:approx:R}
\end{eqnarray}
\end{lemma}
Equations (\ref{eq:approx:E}) and  (\ref{eq:approx:F}) state approximation bounds. Equations (\ref{eq:approx:E}) state that $(1+\varepsilon_1)^{-1}\CE$ bounds below $E$ which bounds below $\CE$. Equations (\ref{eq:approx:F}) state that $\CF$ bounds below $F$ which bounds below $(1+\varepsilon_2)\CF$. Crucially, whatever the sizes of $\CE$ and $\CF$, there exist such approximations $E$ and $F$ that are $O((1/\varepsilon_1)^{d-1})$ and $O((1/\varepsilon_2)^{d-1})$ sized  \cite{papadimitriou2000approximability}, yielding the approximation scheme below. 

\begin{theorem}[Approximation Scheme for MO-CR]\label{th:approx}
Given a compact MOG of representation length $L$, precisions $\varepsilon_1,\varepsilon_2>0$  and two algorithms  to compute approximations $E$ of $\CE$ and $F$ of $\CF$ in the sense of Equations (\ref{eq:approx:E}) and (\ref{eq:approx:F}) that take time $\theta_{\CE}(\varepsilon_1,L)$ and $\theta_{\CF}(\varepsilon_2,L)$, one can approximate $R[\CE,\CF]$ in the sense of Equation (\ref{eq:approx:R}) in time:
$$O\left(\theta_{\CE}(\varepsilon_1,L) \quad+\quad \theta_{\CF}(\varepsilon_2,L) \quad+\quad {(\varepsilon_1 \varepsilon_2)^{-(d-1)(2d-1)}}\right)$$
\end{theorem}
For MO graphical games, Phase 1 could be instantiated with approximate junction-tree algorithms on MO graphical models \cite{dubus2009multiobjective}. For MO symmetric action-graph games, in the same fashion, one could generalize existing algorithms \cite{jiang2007computing}. More generally, for $\WST[\CE]$ and $\CF$, one can also use meta-heuristics with experimental guarantees.

\section{Prospects}

Multi-objective games can be used as a behaviourally more realistic framework to model a wide set of games occurring in business situations ranging from carpooling websites to combinatorial auctions.
Also, studying the efficiency of MO generalizations of routing or Cournot-competitions \cite{guo2005price} could provide realistic economic insights.

%

\newpage

\bibliographystyle{splncs}
\bibliography{main}

\begin{thebibliography}{10}
\providecommand{\url}[1]{\texttt{#1}}
\providecommand{\urlprefix}{URL }

\bibitem{globalissuestobacco}
{G}lobal {I}ssues: {T}obacco. http://www.globalissues.org/article/533/tobacco
  (2014)

\bibitem{adam1776inquiry}
Adam, S.: An inquiry into the nature and causes of the wealth of nations. Edwin
  Cannan's annotated edition  (1776)

\bibitem{aland2006exact}
Aland, S., Dumrauf, D., Gairing, M., Monien, B., Schoppmann, F.: Exact price of
  anarchy for polynomial congestion games. In: STACS 2006, pp. 218--229.
  Springer (2006)

\bibitem{awerbuch2005price}
Awerbuch, B., Azar, Y., Epstein, A.: The price of routing unsplittable flow.
  In: Proceedings of the 37th annual ACM symposium on Theory of computing. pp.
  57--66. ACM (2005)

\bibitem{blackwell1956analog}
Blackwell, D., et~al.: An analog of the minimax theorem for vector payoffs.
  Pacific Journal of Mathematics  6(1),  1--8 (1956)

\bibitem{christodoulou2005price}
Christodoulou, G., Koutsoupias, E.: The price of anarchy of finite congestion
  games. In: Proceedings of the 37th annual ACM symposium on Theory of
  computing. pp. 67--73. ACM (2005)

\bibitem{daskalakis2011connectivity}
Daskalakis, C., Dimakis, A.G., Mossel, E., et~al.: Connectivity and equilibrium
  in random games. The Annals of Applied Probability  21(3),  987--1016 (2011)

\bibitem{dresher1970probability}
Dresher, M.: Probability of a pure equilibrium point in n-person games. Journal
  of Combinatorial Theory  8(1),  134--145 (1970)

\bibitem{dubus2009multiobjective}
Dubus, J.P., Gonzales, C., Perny, P.: Multiobjective optimization using gai
  models. In: IJCAI. pp. 1902--1907 (2009)

\bibitem{guo2005price}
Guo, X., Yang, H.: The price of anarchy of cournot oligopoly. In: Internet and
  Network Economics, pp. 246--257. Springer (2005)

\bibitem{jiang2007computing}
Jiang, A.X., Leyton-Brown, K.: Computing pure nash equilibria in symmetric
  action graph games. In: AAAI. vol.~1, pp. 79--85 (2007)

\bibitem{koutsoupias1999worst}
Koutsoupias, E., Papadimitriou, C.: Worst-case equilibria. In: STACS 99. pp.
  404--413. Springer (1999)

\bibitem{madeley1999big}
Madeley, J.: Big business, poor peoples: the impact of transnational
  corporations on the world's poor. Palgrave Macmillan (1999)

\bibitem{marx1867kapital}
Marx, K.: Das Kapital. Verlag von Otto Meisner (1867)

\bibitem{monderer1996potential}
Monderer, D., Shapley, L.S.: Potential games. Games and economic behavior
  (1996)

\bibitem{papadimitriou2000approximability}
Papadimitriou, C.H., Yannakakis, M.: On the approximability of trade-offs and
  optimal access of web sources. In: Foundations of Computer Science, 2000.
  Proceedings. 41st Annual Symposium on. pp. 86--92. IEEE (2000)

\bibitem{patrone2007multicriteria}
Patrone, F., Pusillo, L., Tijs, S.: Multicriteria games and potentials. Top
  (2007)

\bibitem{rinott2000number}
Rinott, Y., Scarsini, M.: On the number of pure strategy {N}ash equilibria in
  random games. Games and Economic Behavior  33(2),  274--293 (2000)

\bibitem{rosenthal1973class}
Rosenthal, R.W.: A class of games possessing pure-strategy {N}ash equilibria.
  International Journal of Game Theory  2(1),  65--67 (1973)

\bibitem{roughgarden2009intrinsic}
Roughgarden, T.: Intrinsic robustness of the price of anarchy. In: Proceedings
  of the forty-first annual ACM symposium on Theory of computing. pp. 513--522.
  ACM (2009)

\bibitem{roughgarden2007introduction}
Roughgarden, T., Tardos, E.: Introduction to the inefficiency of equilibria.
  Algorithmic Game Theory  17,  443--459 (2007)

\bibitem{shapley1959equilibrium}
Shapley, L.S.: Equilibrium points in games with vector payoffs. Naval Research
  Logistics Quarterly  6(1),  57--61 (1959)

\bibitem{sloan2004price}
Sloan, F.A.: The price of smoking. MIT press (2004)

\bibitem{walras1896elements}
Walras, L.: {\'E}l{\'e}ments d'{\'e}conomie politique pure. F. Rouge (1896)

\bibitem{world2011report}
WHO: {WHO} report on the global tobacco epidemic, 2011: warning about the
  dangers of tobacco: executive summary  (2011)

\end{thebibliography}

\newpage
\appendix

\section{Why Smoothness will not work on multi-objective games}\label{app:smooth}

Most single-objective price of anarchy analytic results rely on a smoothness-analysis \cite{roughgarden2009intrinsic}. A crucial step for ``smoothness'' is to sum the best response inequalities: For a single-objective game and an equilibrium $a\in\PE$, from the best-response conditions $\forall i\in N, \forall b^i\in A^i, v^i(a)\geq v^i(b^i,a^{-i})$, one has:
$\forall b\in A, \sum_{i=1}^{n} v^i(a)\geq\sum_{i=1}^{n} v^i(b^i,a^{-i})$.
However, the Pareto-Nash-equilibrium conditions are rather: 
$\forall i\in N, \forall b^i\in A^i, v^i(b^i,a^{-i}) \not\PP v^i(a)$.
As shown in the following counter-example, such $\not\PP$ relations cannot be summed:
$$
\left(\begin{array}{c} 2\\4 \end{array}\right)
\not\PP
\left(\begin{array}{c} 3\\1 \end{array}\right)
\text{ and }
\left(\begin{array}{c} 3\\1 \end{array}\right)
\not\PP
\left(\begin{array}{c} 1\\2 \end{array}\right)
~~~~~~~~\text{ but }~~~~~~~~
\left(\begin{array}{c} 2\\4 \end{array}\right)
+
\left(\begin{array}{c} 3\\1 \end{array}\right)
\PP
\left(\begin{array}{c} 3\\1 \end{array}\right)
+
\left(\begin{array}{c} 1\\2 \end{array}\right)
$$
Consequently, Smoothness-analysis does not encompass Pareto-Nash equilibria, regardless of the efficiency measurement chosen.

\section{Properties of the Multi-objective Coordination Ratio}\label{app:axioms}

The Multi-objective Coordination Ratio fulfils a list of key good properties for the thorough measurement of the multi-objective efficiency of MO games.

\subsection{Worst case guarantee on equilibria outcomes\\ and No dictatorship on efficient outcomes}

Each vectorial efficiency ratio that the MO-CR states, bounds below the efficiency for each equilibrium outcome, compared to an existing efficient outcome:
$$
\forall \rho\in\text{MO-CR}[\CE,\CF],\quad
\forall y\in\CE,\quad
\exists z\in\CF,\quad
y/z\WPP\rho
$$ 
The process of measuring efficiency by MO-CR does not imply any choice in $\CF$ that would impose a point-of-view telling what efficiency should be (e.g. no five-year plans).

\subsection{Multi-objective ratio-scale}

Given $\CE$, $\CF$ and $r\in\REpp$, it holds that:
\begin{eqnarray}
\text{MO-CR}[\CE,\CF] &\quad\subseteq\quad & \REp\label{eq:ratio:6}\\
\text{MO-CR}[\{(0,\ldots,0)\},\CF] & \quad=\quad & \{(0,\ldots,0)\}\label{eq:ratio:7}\\
\text{MO-CR}[r\star\CE,\CF] & \quad=\quad & r\star\text{MO-CR}[\CE,\CF]\label{eq:ratio:8}\\
\text{MO-CR}[\CE,r\star\CF] & \quad=\quad & \text{MO-CR}[\CE,\CF]/r\label{eq:ratio:9}\\
\CE\subseteq\CF & \quad\Leftrightarrow\quad & (1,\ldots,1)\in \text{MO-CR}[\CE,\CF]\label{eq:ratio:10}
\end{eqnarray}

Equation (\ref{eq:ratio:6}) states that the MO-CR is expressed in the multi-objective space. It is worth noting that while $\text{MO-CR}[\CE,\CF]\subseteq[0,1]^d$ is a more classical choice, MO-CR also allows for measurements of over-efficiencies. (E.g. if $\CF$ is a family-car and $\CE$ is a Lamborghini, then there is over-efficiency on the speed objective.)

\hspace*{0.5cm} Equations (\ref{eq:ratio:7}), (\ref{eq:ratio:8}) and (\ref{eq:ratio:9}) state that MO-CR is sensitive on each objective to multiplications of the outcomes. For instance, if $\CE$ is three times better on objective $k$, then so is MO-CR. If there are twice better opportunities of efficiency in $\CF$ on objective $k'$, then MO-CR is one half on objective $k'$. In other words, the efficiency of each objective independently reflects into the MO-CR in a ratio-scale.

\hspace*{0.5cm} If all equilibria outcomes are efficient (i.e. $\CE\subseteq\CF$), then this must imply that according to the MO-CR, the MO game is fully efficient, that is: $(1,\ldots,1)\in \text{MO-CR}[\CE,\CF]$. The MO-CR seems to be the only multi-objective ratio-scale measurement that fulfils Equation (\ref{eq:ratio:10}) while being a worst case guarantee on equilibria outcomes with no dictatorship on what efficiency should be.

\hspace*{0.5cm} It is also worth noting that MO-CR is MO-monotonic with respect to $\CE$ and $\CF$. For $X,Y\subseteq\REp$, let $X\unrhd Y$ denote $Y\subseteq\CC(X)$ where $\CC(X)=\cup_{x\in X}\CC(x)$ (i.e. $X$ dominates $Y$). Then it holds that:
\begin{eqnarray}
\CE\unrhd\CE' &\quad\Rightarrow\quad & \text{MO-CR}[\CE,\CF]\unrhd\text{MO-CR}[\CE',\CF]\\
\CF\unrhd\CF' &\quad\Rightarrow\quad & \text{MO-CR}[\CE,\CF']\unrhd\text{MO-CR}[\CE,\CF]
\end{eqnarray}

\section{Bounded Rationality in Tobacco Consumption}\label{app:tabac}

%
According to the \textit{intrinsic theory of value} \cite{adam1776inquiry}, the value of a cigarette objectively amounts to the quantities of raw materials used for its production, or is the combination of the labour times put into it \cite{marx1867kapital}. However, each economic agent needs to keep the freedom to evaluate and act how he pleases, in order to keep his good will and some economic efficiency, as observed in the end of the Soviet Union.
According to the \textit{subjective theory of value} \cite{adam1776inquiry}, the value of a cigarette amounts to the price an agent is willing to pay for it. Since the consumers value the product, then the industry creates value \cite{walras1896elements}. 
However, this disregards what the disastrous consequence is on life expectancy, belittles 7.500.000 deaths-per-year and emphasizes the bounded rationality of behaviours. 
 While for some health economists, consuming a cigarette is a rational choice, as one values pleasure more than life expectancy, for others, consumers are stuck into addiction before becoming adults.  The truth is likely between these two extreme points of view \cite{sloan2004price}: Economic agents discount the future at a rate of 6\% per-year, hence a day of life in 40 years is valued 10 times less than now, leading to overweighting the actual smoking pleasure and to irrational behaviours with respect to preferences over a full lifetime. Agents behave according to objectives (e.g. addictive satisfaction) that they would avoid if they had the full experience of their lifetime (e.g. a lung cancer with probability $1/2$) and a sufficient will (e.g. quit smoking). Time discounting also explains other non-sustainable behaviours like over-fishing catastrophes.

\section{Proofs}\label{app:proof}

\subsection{Phase 1 for Normal Forms, Correctness of Algorithm \ref{alg:givenEFoutputMOPoA:polynomial} and Theorem \ref{th:master}}
\label{sec:mopoa:poly}

Phase 1 is easy, if the MOG is given in normal form. The MOG is made of the MO evaluations of each agent on each action-profile, that is: $O(n\alpha^n)$ vectors. Hence, the computation of $u(A)$ requires
for each $a\in A$, the addition of $n$ vectors. Therefore, the computation of $u(A)$ takes time $O(\alpha^n n d)$ (linear in the size of the input) and yields $O(\alpha^n)$ vectors. The computation of $\CF=\EFF[u(A)]$ given $u(A)$ takes time $O(|u(A)|^2 d)=O(\alpha^{2n} d)$. To conclude, the computation of $\CF$ takes time $O(n \alpha^n d + \alpha^{2n} d)$, which is polynomial (quadratic) in the size of the input. If $d=2$, this can be significantly lowered to $O(n \alpha^n d + \alpha^{n}\log_2(\alpha^n)d)=O(n\alpha^n \log_2(\alpha) )$.

\hspace*{0.7cm}The computation of $\WST[\CE]$ can be achieved by first computing $\PE$. For this purpose, for each agent $i\in N$ and each adversary action profile $a^{-i}\in A^{-i}$, one has to compute which individual actions give a Pareto-efficient evaluation in $v^i(A^i,a^{-i})$, in order to mark which action-profiles can be a $\PE$ from $i$'s point of view. Overall, computing $\PE$ takes time $O(n \alpha^{n-1} \alpha^2 d)$. Using back $u(A)$, computing $\CE=u(\PE)$ is straightforward. Again,  the computation of $\WST[\CE]$ given $\CE$ takes time $O(|\CE|^2 d)=O(\alpha^{2n} d)$. To sum up, the computation of $\WST[\CE]$ takes time $O(n \alpha^{n+1} d + \alpha^{2n} d)$. If $d=2$, this lowers to $O(n\alpha^n \log_2(\alpha))$.\\

\hspace*{0.7cm}In order to compute $\text{MO-CR}=\EFF[R[\WST[\CE],\CF]]$,
let us study the structure of $\bigcap_{y\in\WST[\CE]}\bigcup_{z\in\CF}\CC(y/z)$, by restricting a set-algebra to the following objects:
\begin{definition}[Cone-Union]
\label{def:coneunion}
For a set of vectors $X\subseteq\REp$, the Cone-Union $\CC(X)$ is:
$$
\CC(X)
~~=~~\bigcup_{x\in X}\CC(x)
~~~~=\{y\in\REp ~~|~~ \exists x\in X, x\WPP y\}
$$ 
Let $\CC$ denote the set of all cone-unions of $\REp$.
\end{definition}
To define an algebra on $\CC$, one can supply $\CC$ with $\cup$ and $\cap$.
\begin{lemma}[On the Set-Algebra $(\CC,\cup,\cap)$]\label{prop:algebra}~\\
Given two descriptions of cone-unions $X^1,X^2\subseteq\REp$, we have:
$$\CC(X^1)\cup\CC(X^2)=\CC( X^1\cup X^2 )$$
Given two descriptions of cones $x^1,x^2\in\REp$, we have:
$$\CC(x^1)\cap\CC(x^2)=\CC(x^1\wedge x^2)$$
where $x^1\wedge x^2\in\REp$ is: $\forall k\in\OBJ, (x^1\wedge x^2)_k=\min\{x^1_k,x^2_k\}$.\\
Given two descriptions of cone-unions $X^1,X^2\subseteq\REp$, we have:
\begin{eqnarray*}
\CC(X^1)\cap\CC(X^2)
&=&\left(\cup_{x^1\in X^1}\CC(x^1)\right)\cap\left(\cup_{x^2\in X^2}\CC(x^2)\right)\\
&=&\bigcup_{(x^1,x^2)\in X^1\times X^2}\CC(x^1)\cap\CC(x^2)\\
&=&\bigcup_{(x^1,x^2)\in X^1\times X^2}\CC(x^1\wedge x^2)\\
&=&\CC( X^1\wedge X^2 )
\end{eqnarray*}
where $X^1\wedge X^2=\{x^1\wedge x^2~|~x^1\in X^1,~x^2\in X^2\}\subseteq\REp$.\\
Therefore, $(\CC,\cup,\cap)$ is stable, and then is a set-algebra.
\end{lemma}
\begin{proof}
The three properties derive from set calculus.
\end{proof}
The main consequence of Lemma \ref{prop:algebra} is that  $R[\WST[\CE],\CF]=\cap_{y\in\WST[\CE]}\cup_{z\in\CF}\CC(y/z)$ is a cone-union. Moreover, one can do the development for $\cap_{y\in\WST[\CE]}\cup_{z\in\CF}\CC(y/z)$ within the cone-unions, using distributions and developments.

\begin{remark}\label{rk:app:cone}
For a finite set $X\subseteq\REp$, we have: $\CC(X)=\CC(\EFF[X])$.
\end{remark}
\begin{proof}
Firstly, we prove $\CC(X)\subseteq\CC(\EFF[X])$.
If $y\in\CC(X)$, then there exists $x\in X$ such that $x\WPP y$. There are two cases, $x\in\EFF[X]$ and $x\not\in\EFF[X]$. If $x\in\EFF[X]$, then $y\in\CC(\EFF[X])$, by definition of a cone-union.
Otherwise, if $x\not\in\EFF[X]$, then there exists $z\in X$ such that $z\PP x$. And since $X$ is finite, we can find such a $z$ in $\EFF[X]$, by iteratively taking $z'\PP z$ until $z\in\EFF[X]$, which will happen because $X$ is finite and $\PP$ is transitive and irreflexive. Hence, there exists $z\in\EFF[X]$ such that $z\PP x\WPP y$ and then $z\WPP y$. Consequently, $y\in\CC(\EFF[X])$, by definition of a cone-union.

Conversely, $Y\subseteq X\Rightarrow \CC(Y)\subseteq\CC(X)$  proves $\CC(\EFF[X])\subseteq\CC(X)$.
\end{proof}

As a consequence of Remark \ref{rk:app:cone}, for $x\in\REp$, a simple cone $\CC(x)$ is fully described by its summit $x$. The main consequence of this remark is that $\CC(X)$ can be fully described and represented by $\EFF[X]$.
For instance, since $R[\WST[\CE],\CF]$ is a cone-union (thanks to Lemma \ref{prop:algebra}), and since $\text{MO-CR}=\EFF[R[\WST[\CE],\CF]]$ (by definition of the MO-CR), then $R[\WST[\CE],\CF]$ is fully represented (as a cone-union) by the MO-CR, which means that $R[\WST[\CE],\CF]=\CC(\text{MO-CR})$.\\


 Recall that $q=|\WST[\CE]|$ and $m=|\CF|$. In this subsection, we also denote $\WST[\CE]=\{y^1,\ldots,y^q\}$ and $\CF=\{z^1,\ldots,z^m\}$.
Let $\CA_q^m$ denote the set of functions $\pi$ from $\{1,\ldots,q\}$ to $\{1,\ldots,m\}$. (We have: $|\CA_q^m|=m^q$.)
\begin{corollary}[The cone-union of MO-CR]~\\
Given $\WST[\CE]=\{y^1,\ldots,y^q\}$ and $\CF=\{z^1,\ldots,z^m\}$, we have:
$$R[\WST[\CE],\CF]=\bigcup_{\pi\in\CA_q^m}\bigcap_{t=1}^{q} \CC(y^t / z^{\pi(t)})$$
and therefore:
$$\text{MO-CR}=\EFF\left[\left\{\bigwedge\nolimits_{t=1}^{q}y^t / z^{\pi(t)}~~|~~\pi\in\CA_q^m\right\}\right]$$
\end{corollary}
\begin{proof}
For the first statement, just think to a development. We write down $R[\WST[\CE],\CF]=\cap_{y\in\WST[\CE]}\cup_{z\in\CF}\CC(y/z)$ into the layers just below. There is one layer per $y^t$ in $\WST[\CE]=\{y^1,\ldots,y^t,\ldots,y^q\}$:
$$
\begin{array}{ccccccccccl}
& ( & \CC(\frac{y^1}{z^1}) & \cup & \CC(\frac{y^1}{z^2}) & \cup & \ldots & \cup & \CC(\frac{y^1}{z^m})& )&\text{layer 1}\\
\bigcap & ( & \CC(\frac{y^2}{z^1}) & \cup & \CC(\frac{y^2}{z^2}) & \cup & \ldots & \cup & \CC(\frac{y^2}{z^m})& )&\text{layer 2}\\
&&&&&\vdots\\
\bigcap & ( & \CC(\frac{y^q}{z^1}) & \cup & \CC(\frac{y^q}{z^2}) & \cup & \ldots & \cup & \CC(\frac{y^q}{z^m})& )&\text{layer q}
\end{array}
$$
Imagine the simple cones as vertices and imagine edges going from each vertex of layer $t$ to each vertex of the next layer $(t+1)$. The development into a union outputs as many intersection-terms as paths from the first layer to the last one. 
Let the function $\pi:\{1,\ldots,q\}\rightarrow\{1,\ldots,m\}$ denote a path from layer $1$ to layer $q$, where $\pi(t)$ is the vertex chosen in layer $t$.
Consequently, in the result of the development into an union, each term is an intersection $\bigcap_{t=1}^{q} \CC(y^t / z^{\pi(t)})$.

The second statement results from the first statement, Lemma \ref{prop:algebra} and Remark \ref{rk:app:cone}.
\begin{eqnarray*}
R[\WST[\CE],\CF]
&=&\bigcup_{\pi\in\CA_q^m}\bigcap_{t=1}^{q} \CC(y^t / z^{\pi(t)})\\
&=&\bigcup_{\pi\in\CA_q^m} \CC\left(\bigwedge_{t=1}^{q} y^t / z^{\pi(t)}\right)\\
&=&\CC\left(\left\{\bigwedge\limits_{t=1}^{q}y^t / z^{\pi(t)}~~|~~\pi\in\CA_q^m\right\}\right)
\end{eqnarray*}
That $R[\WST[\CE],\CF]=\CC(\text{MO-CR})$ concludes the proof.
\end{proof}

Ultimately, this proves the \textbf{correctness} of Algorithm \ref{alg:givenEFoutputMOPoA:polynomial} for the computation of MO-CR, given $\WST[\CE]=\{y^1,\ldots,y^q\}$ and $\CF=\{z^1,\ldots,z^m\}$. It consists in the iterative development of the intersection $R(\CE,\CF)$, which can be seen as dynamic programming on the paths of the layer graph. 
For $k\in\{1,\ldots,q\}$, we denote $D^t$ the description of the cone-union corresponding to the intersection: 
$$\CC(D^t)=\cap_{l=1}^{t} \cup_{z\in\CF} \CC(y^l/z)$$
Recursively, for $t>1$,  $\CC(D^t)=\CC(D^{t-1})~\cap~(\cup_{z\in\CF}~\CC(y^{t}/z))$.
From Lemma \ref{prop:algebra} and Remark \ref{rk:app:cone}, in order to develop, we 
then have to iterate the following:
$$
D^t=\EFF[\{\rho ~\wedge~ (y^t/z)~~|~~\rho\in D^{t-1},~~z\in\CF\}]
$$

We now proceed with the \textbf{time complexity} of Algorithm \ref{alg:givenEFoutputMOPoA:polynomial}. At first glance, since there are $m^q$ paths in the layer graph, then there are $O(m^q)$ elements in MO-CR. Fortunately, they are much less, because we have:
\begin{theorem}[MO-CR is polynomially-sized]~\\
\label{th:mopoa:poly}
Given a MOG and denoting $d=|\OBJ|$, $q=|\WST[\CE]|$ and $m=|\CF|$, we have:
$$|\text{MO-CR}|\leq (qm)^{d-1}$$
\end{theorem}
\begin{proof}
Given $\rho\in\text{MO-CR}$, for some $\pi\in\CA_q^m$, we have $\rho=\bigwedge\nolimits_{t=1}^{q}y^t / z^{\pi(t)}$, and then $\forall k\in\OBJ, \rho_k=\min_{t=1\ldots q}\{y^t_k / z^{\pi(t)}_k\}$. Therefore, $\rho_k$ is exactly realized by the $k$th component of at least one cone summit $y^t / z^{\pi(t)}$ in the layer graph (that is a vertex in the layer-graph above). Consequently, there are at most as many possible values for the $k$th component of $\rho$, as the number of vertices in the layer graph, that is $qm$. This holds for the $d$ components of $\rho$; hence there are at most $(qm)^d$ vectors in MO-CR. More precisely, by Lemma \ref{lem:eff} (below), since MO-CR is an efficient set, then there are at most $(qm)^{d-1}$ vectors in MO-CR.
\end{proof}
\begin{lemma}\label{lem:eff}
Let $Y\subseteq\REp$ be a set of vectors, with at most $M$  values on each component:
$$|~\EFF[Y]~|\leq M^{d-1}$$
\end{lemma}
\begin{proof}  At most $M^{d-1}$ valuations are realized on the $d-1$ first components. If you fix the $d-1$ first components, there is at most one Pareto-efficient vector which maximizes the last component.
\end{proof}

In Algorithm \ref{alg:givenEFoutputMOPoA:polynomial}, there are $\Theta(q)$ steps. At each step $t$, from Theorem \ref{th:mopoa:poly}, we know that $|D^{t-1}|\leq (qm)^{d-1}$. Hence, $|\{\rho ~\wedge~ (y^t/z)~~|~~\rho\in D^{t-1},~~z\in\CF\}|\leq q^{d-1} m^d$, and the computation of the efficient set $D^t$ requires time $O((q^{d-1} m^d)^2 d)$. 
However, by using an insertion process, since there are at most $B=|D^{t}|\leq (qm)^{d-1}$ Pareto-efficient vectors at each insertion, then we only need $O(q^{d-1} m^d \times (qm)^{d-1})$ Pareto-comparisons. If $d=2$, time lowers to $O(q^{d-1} m^d\log_2(q m) d^2)=O(q m^2\log_2(q m))$.

Ultimately, Algorithm \ref{alg:givenEFoutputMOPoA:polynomial} takes $q$ steps and then time $O(q (q^{d-1} m^d)(qm)^{d-1} d)=O((qm)^{2d-1}d)$. If $d=2$, this lowers to $O((q m)^2\log_2(q m))$.

\subsection{Approximations: Proof of Lemma \ref{lem:approx} and Theorem \ref{th:approx}}

\begin{proof}
(1) First, let us  show  $R[E,F]\subseteq R[\WST[\CE],\CF]$. Let $\rho'$ be a ratio of $R[E,F]$ and let us show that:
$$
\forall y\in \WST[\CE],~~
\exists z\in \CF,~~
\text{ s.t.: } y\WPP\rho'\star z
$$
Take $y\in\WST[\CE]$. From Equation (\ref{eq:approx:E}) (first condition), there is a $y'\in E$ such that $y\WPP y'$.
From Definition \ref{def:mopoa}, there is a $z'$ such that $y'\WPP \rho'\star z'$. From Equation (\ref{eq:approx:F}) on $z'$ (first condition), there exists $z\in\CF$ such that $z'\WPP z$. Recap: $y\WPP y'\WPP \rho'\star z'\WPP \rho'\star z$.

(2) Then, let $\rho$ be a ratio of $R[\WST[\CE],\CF]$,
and let us  show that $\rho'=(1+\varepsilon_1)^{-1}(1+\varepsilon_2)^{-1}\rho$~~ is in $R[E,F]$, that is:
$$
\forall y'\in E,~~
\exists z'\in F,~~
(1+\varepsilon_1) y'\WPP (1+\varepsilon_2)^{-1} \rho\star z'
$$
Take an element $y'$ of $E$.
From Equation (\ref{eq:approx:E}) (second condition), there is $y\in\WST[\CE]$ such that $(1+\varepsilon_1)y'\WPP y$. 
From Definition \ref{def:mopoa}, there is $z\in\CF$ such that $y\WPP\rho\star z$.
From Equation (\ref{eq:approx:E}) on $z$ (second condition), there exists $z'\in F$ s.t. $z\WPP (1+\varepsilon_2)^{-1} z'$.
Recap: $(1+\varepsilon_1)y'\WPP y\WPP \rho\star z\WPP (1+\varepsilon_2)^{-1} \rho\star z'$.
\end{proof}

\begin{proof}[Theorem \ref{th:approx}]
For the first claim, since Algorithm \ref{alg:givenEFoutputMOPoA:polynomial}, given $E$ and $F$, outputs the MO-PoA corresponding to $R[E,F]$, by Theorem \ref{th:approx}, Algorithm \ref{alg:givenEFoutputMOPoA:polynomial} outputs an $((1+\varepsilon_1)(1+\varepsilon_2))$-covering of $R(\WST[\CE],\CF)$.

For the second claim, from Lemma \ref{lem:approx}, applying Algorithm \ref{alg:givenEFoutputMOPoA:polynomial} on $E$ and $F$ outputs an $((1+\varepsilon_1)(1+\varepsilon_2))$-covering of $R(\WST[\CE],\CF)$.
Moreover, since we have $|E|=O((1/\varepsilon_1)^{d-1})$ and $|F|=O((1/\varepsilon_2)^{d-1})$, Algorithm \ref{alg:givenEFoutputMOPoA:polynomial} takes time $O\left(d/(\varepsilon_1\varepsilon_2)^{(d-1)(2d-1)}\right)$.
\end{proof}

\end{document}